\documentclass[11pt]{amsart}

\usepackage[T1]{fontenc}
\usepackage[utf8]{inputenc}
\usepackage[english]{babel}
\usepackage{relsize}
\usepackage{verbatim}
\usepackage{mathtools,amssymb,amsthm}
\usepackage[normalem]{ulem}

\usepackage[maxnames=50,giveninits=true]{biblatex}
\usepackage{csquotes}
\usepackage[shortlabels]{enumitem}
\usepackage{multirow}
\setlist{nolistsep}
\setenumerate{label={{\normalfont({\roman*})}},leftmargin=6mm,itemsep=3pt,topsep=3pt}
\setitemize{label={$\vcenter{\hbox{\tiny$\bullet$}}$},leftmargin=6mm,itemsep=3pt,topsep=3pt}
\usepackage{caption,subcaption}

\usepackage[margin=3cm,includefoot,footskip=30pt]{geometry}

\usepackage{array}
\newcolumntype{C}[1]{>{\arraybackslash$}p{#1}<{$}}

\usepackage{titlesec}
\titleformat{\section}{\normalfont\rmfamily\Large\bfseries}{\thesection.}{14pt}{}
\titlespacing{\section}{0pt}{14pt}{7pt}

\titleformat{\subsection}{\normalfont\rmfamily\large\bfseries}{\thesubsection.}{12pt}{}
\titlespacing{\subsection}{0pt}{12pt}{6pt}

\renewcommand{\emptyset}{\varnothing}

\renewcommand{\leq}{\leqslant}

\renewcommand{\geq}{\geqslant}

\usepackage[algoruled,linesnumbered,vlined]{algorithm2e}
\usepackage{algpseudocode}

\usepackage{ifthen}
\usepackage{tikz}
\usetikzlibrary{matrix,fit,calc,positioning}
\tikzstyle{every picture}=[line join=round,line cap=round,line width=.75pt,every label/.append style={font=\small},label distance=-1.5pt]
\tikzstyle{every node}=[font=\small]
\tikzset{>=stealth}

\newcommand{\Q}{\mathbb Q}
\newcommand{\R}{\mathbb R}
\newcommand{\N}{\mathbb N}

\newcommand{\aff}{\mathrm{aff}}
\newcommand{\conv}{\mathrm{conv}}
\newcommand{\col}{\mathrm{col}}

\newcommand{\stab}{\mathrm{STAB}}
\newcommand{\pr}{\mathrm{Pr}}

\newcommand{\rk}{\mbox{rk}}

\newcommand{\calB}{\mathcal{B}}

\newtheorem{prop}{Proposition}
\newtheorem{lem}[prop]{Lemma}
\newtheorem{obs}[prop]{Observation}
\newtheorem{conj}[prop]{Conjecture}
\newtheorem{thm}[prop]{Theorem}
\newtheorem{cor}[prop]{Corollary}

\theoremstyle{definition}

\newcommand{\restr}[1]{\big|_{#1}}

\newcommand{\op}{product}

\makeatletter
\renewcommand\paragraph{\@startsection{paragraph}{4}{\z@}%
{3.25ex \@plus1ex \@minus.2ex}{-.2em}%
{\normalfont\normalsize\bfseries}}
\makeatother

\begin{document}

\title{
Slack matrices, $k$-products, and $2$-level polytopes}
\author{Manuel Aprile, Michele Conforti, Yuri Faenza,\\ Samuel Fiorini, Tony Huynh, Marco Macchia}
\date{}
%
%
\maketitle


\begin{abstract}
    In this paper, we study algorithmic questions concerning products of matrices and their consequences for recognition algorithms for polyhedra. 
    
    The \emph{$1$-\op} of matrices $S_1 \in \R^{m_1 \times n_1}, S_2 \in \R^{m_2 \times n_2}$ is a matrix in $\R^{(m_1+m_2) \times (n_1n_2)}$ whose columns are the concatenation of each column of $S_1$ with each column of $S_2$. The \emph{$k$-\op{}} generalizes the $1$-\op{}, by taking as input two matrices $S_1, S_2$ together with $k-1$ special rows of each of those matrices, and outputting a certain composition of $S_1,S_2$. 
Our first result is a polynomial time algorithm for the following problem:  given a matrix $S$, is $S$ a $k$-\op{} of some matrices, up to permutation of rows and columns? Our algorithm is based on minimizing a symmetric submodular function that expresses mutual information from information theory.  

Our study is motivated by a close link between the 1-\op{} of matrices and the Cartesian product of polytopes, and more generally between the $k$-\op{} of matrices and the glued product of polytopes. These connections  rely on the concept of slack matrix, which gives an algebraic representation of classes of affinely equivalent polytopes. The \emph{slack matrix recognition problem} is the problem of determining whether a given matrix is a slack matrix.  This is an intriguing problem whose complexity is unknown. Our algorithm reduces the problem to instances which cannot be expressed as $k$-products of smaller matrices. 

In the second part of the paper, we give a combinatorial interpretation of $k$-\op{}s for two well-known classes of polytopes: 2-level matroid base polytopes and stable set polytopes of perfect graphs. We also show that the slack matrix recognition problem is polynomial-time solvable for such polytopes. Those two classes are special cases of $2$-level polytopes, for which we conjecture that the slack matrix recognition problem is polynomial-time solvable. 
\end{abstract}




\section{Introduction}\label{sec:intro}


Determining if an object can be decomposed as the `product' of simpler objects is a ubiquitous theme in mathematics and computer science.  For example, every integer has a unique factorization into primes, and every finite abelian group is the direct sum of cyclic groups.  Moreover, algorithms to efficiently \emph{find} such `factorizations' are widely studied, since many algorithmic problems are easy on indecomposable instances. In this paper, our objects of interest are matrices and polytopes.

For $\ell \in \N$, we let $[\ell]:=\{1, \dots, \ell\}$.  For a matrix $S$, we let $S^\ell$ be the $\ell$-th column of $S$. 
The \emph{$1$-\op} of $S_1 \in \R^{m_1 \times n_1}$ and $S_2 \in \R^{m_2 \times n_2}$ is the matrix $S_1 \otimes S_2 \in \R^{(m_1+m_2) \times (n_1n_2)}$ such that for each $j \in  [n_1\cdot n_2]$, 
$$
(S_1 \otimes S_2)^{j} := \begin{pmatrix}
S_1^{k} \\ S_2^{\ell}
\end{pmatrix},
$$
where $k \in [n_1]$ and $\ell \in [n_2]$ satisfy $j=(k-1)n_2+\ell$. For example, 
$$
\begin{pmatrix} 1 &0 \end{pmatrix} \otimes \begin{pmatrix} 0 \end{pmatrix} =
\begin{pmatrix}
1 & 0\\
0 & 0
\end{pmatrix}
, \qquad
\begin{pmatrix} 1 &0 
\\ 2 & 3\end{pmatrix} \otimes \begin{pmatrix} 1 &0&0 \\ 0  & 1 & 1 \\ \end{pmatrix}=\begin{pmatrix} 1 &1 &1 &0&0 &0 
\\ 2 & 2& 2& 3& 3& 3\\ 1 &0&0 & 1&0&0\\ 0  & 1 & 1 & 0  & 1 & 1 \end{pmatrix}.
$$
Two matrices are \emph{isomorphic} if one can be obtained from the other by permuting rows and columns. A matrix $S$ is a \emph{1-\op} if there exist two non-empty matrices $S_1$  and $S_2$ such that $S$ is isomorphic to $S_1\otimes S_2$. 

Our motivation for studying the 1-\op{} is geometric.  If $P_1 \subseteq \R^{d_1}$ and $P_2 \subseteq \R^{d_2}$ are polytopes, then their \emph{Cartesian product} is the polytope $$P_1 \times P_2 := \{(x_1,x_2) \in \R^{d_1} \times \R^{d_2} \mid x_1 \in P_1,\ x_2 \in P_2\}.$$ 

It turns out that if we represent a polytope via its \emph{slack matrix}, then the Cartesian product corresponds exactly to the 1-\op{}, see Lemma~\ref{lem:1-sum_slack}.  We now define slack matrices of polytopes.  

Let $P = \conv(\{v_1,\dots, v_n\}) = \{x\in \mathbb{R}^d \mid Ax \leq b, \, A^= x= b^= \}$, where $\{v_1,\dots,v_n\} \subseteq \R^d$, $A \in \R^{m \times d}$ and $b \in \R^m$. The \emph{slack matrix} associated to these descriptions of $P$ is the matrix $S \in \R^{m \times n}_+$ with $S_{i,j} := b_i-A_i v_j$ for $i\in [m]$ and $j\in [n]$.  That is, $S_{i,j}$ is the slack of point $v_j$ with respect to the inequality $A_i x \leq b_i$.
Throughout the paper, we assume that the set $\{v_1,\dots, v_n\}$ is minimal. That is, $v_1,\dots, v_n$ are the vertices of $P$.

 Slack matrices were introduced in a seminal paper of Yannakakis~\cite{yannakakis1991expressing}, as a tool for reasoning about the extension complexity of polytopes (see also~\cite{conforti2013extended}). They are fascinating objects, which are still not fully understood. For instance, given a matrix $S \in \R^{m \times n}_+$, the complexity of determining whether $S$ is the slack matrix of some polytope is open. In \cite{gouveia2013nonnegative}, the problem has been shown to be equivalent to the Polyhedral Verification Problem (see~\cite{kaibel2003some}): given a vertex description of a polytope $P$, and an inequality description of a polytope $Q$, determine whether $P=Q$. 
 
 A more general operation than the Cartesian product was introduced in \cite{margot1995composition} (see also \cite{tiwary2020extension}) for 0/1 polytopes. Here, we extend the definition from~\cite{margot1995composition, tiwary2020extension} to general polytopes.  For a polytope $P$, we let $V(P)$ be its set of vertices. Given coordinates $x_1,\dots,x_{k-1}$ of $P_1$, and coordinates $y_1,\dots,y_{k-1}$ of $P_2$, the \emph{glued product} of $P_1, P_2$ with respect to coordinates $(x_1,\dots,x_{k-1}$, $y_1,\dots, y_{k-1})$ is defined as:
\[
P_1\times_k P_2:=\conv\{(x,y): x\in V(P_1), y\in V(P_2), x_i=y_i \,\forall\, i \in [k-1]\}.
\]


The glued product is a very general operation, and the inequality description of $P_1 \times_k P_2$ cannot be easily derived from that of $P_1,P_2$. Hence, we restrict the input polytopes so that their glued product has a simpler description. In particular, throughout the paper we assume that the projection of $P_1$ (resp.~$P_2$) over coordinates $x_1,\dots,x_{k-1}$ (resp.~$y_1,\dots,y_{k-1}$) is given by 
$$\conv\{\mathbf{0},e_1,\dots, e_{k-1}\},$$
where $e_i$ is the $(k-1)$-dimensional vector equal to $0$ in all entries, except a $1$ in position $i$.
Moreover, we assume that for $i=1,\dots, k-1$, $x_i\geq 0$ (resp. $y_i\geq 0$)  is facet-defining for $P_1$ (resp.~$P_2$).

 
When those properties hold, we call the corresponding glued product \emph{simplicial}. It is known (see \cite{margot1995composition} or our Lemma \ref{lem:gluedprod}) that if $P$ is a simplicial product $P_1\times_k P_2$, then a description of $P$ is obtained by adding equations $x_i=y_i$ for $i\in [k-1]$ to the Cartesian product $P_1\times P_2$. 

If we represent polytopes via  slack matrices, what operation on matrices corresponds to the simplicial glued product?  The answer is an extension of the 
$1$-\op{} called the \emph{$k$-\op}, which we now describe.

We say that a matrix $S$ has \emph{special rows} $(x_1,\dots,x_{k-1})$ if the columns of $S$, restricted to such rows, consist of all the vectors $\mathbf{0}$, $e_1, \dots, e_{k-1}$, as defined above, possibly repeated.

 We consider two matrices $S_1, S_2$ each with an ordered set of special rows,
 $(x_1,\dots,x_{k-1})$ and $(y_1,\dots, y_{k-1})$ respectively. 
 Given a matrix $S$ and a set of rows $\mathcal R$ of $S$, we let $S- \mathcal R$ be the submatrix of $S$ obtained by removing the rows in $\mathcal R$.  
 
 
Observe that $S_1-\{x_1,\dots,x_{k-1}\}$ can be partitioned into submatrices $S_1[0]$, $S_1[1], \dots, S_1[k-1]$ as follows.  Let $S_1[0]$ be the matrix obtained from $S_1-\{x_1,\dots,x_{k-1}\}$ by restricting to columns $S^\ell$ such that the $\ell$-th entry of $x_i$ is $0$ for all $i \in [k-1]$. For $j \in [k-1]$, let $S_1[j]$ be obtained from $S_1$ by restricting to the columns $S^\ell_1$ such that the $\ell$-th entry of $x_j$ is $1$ and the $\ell$-th entry of $x_i$ is $0$ for $i \neq j$. Notice that all submatrices $S_1[j]$ are non-empty.  Because of this and of the assumption above, $S_1[0]$, $S_1[1], \dots, S_1[k-1]$ gives a partition of $S_1-\{x_1,\dots, x_{k-1}\}$ into non-empty submatrices.  We define $S_2[0]$, $S_2[1], \dots, S_2[k-1]$ analogously.
 
With the previous notation, the \emph{$k$-\op} $S$ of matrices $S_1$ and $S_2$ with special rows $(x_1,\dots,x_{k-1})$ and $(y_1,\dots,y_{k-1})$ respectively, is defined to be:
%
$$(S_1,x_1,\dots,x_{k-1}) \otimes_k (S_2,y_1,\dots,y_{k-1}) :=
\left(\begin{array}{c|c|c|c|c}
 S[0]&  S[1] & S[2] & \cdots & S[k-1] \\
 0 \cdots 0 & 1 \cdots 1 & 0 \cdots 0 & \cdots & 0 \cdots 0  \\
 0 \cdots 0 & 0 \cdots 0 & 1 \cdots 1 & \cdots & 0 \cdots 0 \\ \vdots &\vdots &\vdots &\vdots &\vdots \\
 0 \cdots 0 & 0 \cdots 0 & 0 \cdots 0 & \cdots & 1 \cdots 1
\end{array}
\right)
$$ where, for every $j \in [k-1]\cup\{0\}$, $S[j] :=  S_1[j] \otimes S_2[j]$, each column below $S[0]$ is a zero vector of length $k-1$, and each column below $S[j]$ is the $j$-th standard basis vector in $\mathbb{R}^{k-1}$. Notice that the last $k-1$ rows of $S$ are special rows of $S$. We refer to $S_1$ and $S_2$ as the \emph{factors} of the $k$-product.


Here are examples of a 2-\op{} and a 3-\op{}, where we indicate entries of non-special rows with letters for the sake of clarity.

 \begin{center}
 \begin{tikzpicture}[every left delimiter/.style={xshift=.3em},
     every right delimiter/.style={xshift=-.3em}]
 \matrix(S1)[matrix of math nodes,left delimiter={(},right
 delimiter={)},column sep=-\pgflinewidth,row sep=-\pgflinewidth,
 nodes={anchor=center,align=left,inner sep=1pt,text depth=0},
 row 1/.style={nodes={minimum height=5mm}},row 2/.style={minimum height=5mm},row 3/.style={nodes={minimum height=5mm}}]
 {
 a & c & e & g\\
 b & d & f & h\\
 0 & 0 & 1 & 1\\
 };
 \node [left= of S1.center, xshift=1mm] {$S_1 = $};
 \node [right,xshift=1.5mm,yshift=-.5mm] at (S1-3-4.east) {$\leftarrow x_1$};


 \matrix(S2)[matrix of math nodes,left delimiter={(},right
 delimiter={)},column sep=-\pgflinewidth,row sep=-\pgflinewidth,
 nodes={anchor=center,align=center,inner sep=1pt,text depth=0},
 row 1/.style={nodes={minimum height=5mm}},row 2/.style={minimum height=5mm},row 3/.style={nodes={minimum height=5mm}},xshift=95pt] at (S1-2-4.east)
 {
 i & j & k\\
 0 & 0 & 1\\
 };
 \node [left= of S2.center, xshift=1mm] {$S_2 = $};
 \node [right,xshift=1.5mm,yshift=-.5mm] at (S2-2-3.east) {$\leftarrow y_1$};

 \matrix(S)[matrix of math nodes,left delimiter={(},right
 delimiter={)},column sep=-\pgflinewidth,row sep=-\pgflinewidth,
 nodes={anchor=center,align=left,inner sep=1pt,text depth=0},
 row 1/.style={nodes={minimum height=5mm}},row 2/.style={minimum height=5mm},row 3/.style={nodes={minimum height=5mm}},row 4/.style={nodes={minimum height=5mm}},xshift=180pt] at (S2.east)
 {
  a & a & c & c & e & g\\
 b & b & d & d & f & h \\
  i & j & i & j & k & k  \\
 0 & 0 & 0 & 0 & 1 & 1\\
 };
 \node [left= of S.center, xshift=-3mm] {$(S_1,x_1) \otimes_2 (S_2,y_1) = $};
 \end{tikzpicture}
 \end{center}

\begin{center}
 \begin{tikzpicture}[every left delimiter/.style={xshift=.3em},
     every right delimiter/.style={xshift=-.3em}]
 \matrix(S1)[matrix of math nodes,left delimiter={(},right
 delimiter={)},column sep=-\pgflinewidth,row sep=-\pgflinewidth,
 nodes={anchor=center,align=left,inner sep=1pt,text depth=0},
 row 1/.style={nodes={minimum height=5mm}},row 2/.style={minimum height=5mm},row 3/.style={nodes={minimum height=5mm}}, row 4/.style={nodes={minimum height=5mm}}, xshift=-20pt]
 {
 a & c & e & g\\
 b & d & f & h\\
 0 & 0 & 1 & 0\\
 0 & 0 & 0 & 1\\
 };
 \node [left= of S1.center, xshift=1mm,yshift =2mm] {$S_1 = $};
 \node [right,xshift=1.5mm,yshift=-.5mm] at (S1-3-4.east) {$\leftarrow x_1$};
 \node [right,xshift=1.5mm,yshift=-.5mm] at (S1-4-4.east) {$\leftarrow x_2$};


 \matrix(S2)[matrix of math nodes,left delimiter={(},right
 delimiter={)},column sep=-\pgflinewidth,row sep=-\pgflinewidth,
 nodes={anchor=center,align=center,inner sep=1pt,text depth=0},
 row 1/.style={nodes={minimum height=5mm}},row 2/.style={minimum height=5mm},row 3/.style={nodes={minimum height=5mm}},xshift=95pt] at (S1-2-4.east)
 {
 i & j & k & $\ell$\\
 0 & 1 & 0 & 0\\
 0 & 0 & 1 & 1\\
 };
 \node [left= of S2.center, xshift=1mm] {$S_2 = $};
 \node [right,xshift=1.5mm,yshift=-.5mm] at (S2-2-4.east) {$\leftarrow y_1$};
 \node [right,xshift=1.5mm,yshift=-.5mm] at (S2-3-4.east) {$\leftarrow y_2$};

 \matrix(S)[matrix of math nodes,left delimiter={(},right
 delimiter={)},column sep=-\pgflinewidth,row sep=-\pgflinewidth,
 nodes={anchor=center,align=left,inner sep=1pt,text depth=0},
 row 1/.style={nodes={minimum height=5mm}},row 2/.style={minimum height=5mm},row 3/.style={nodes={minimum height=5mm}},row 4/.style={nodes={minimum height=5mm}},xshift=220pt] at (S2.east)
 {
 a & c & e & g & g  \\
 b & d & f & h & h \\
 i & i & j & k & $\ell$  \\
 0 & 0 & 1 & 0 & 0  \\
 0 & 0 & 0 & 1 & 1 \\
 };
 \node [left= of S.center, xshift=-3mm] {$(S_1,x_1,x_2) \otimes_3 (S_2,y_1,y_2) = $};
 \end{tikzpicture}
 \end{center}

As in the {$1$-\op} case, we say that $S$ is a \emph{$k$-\op{}} if there exist matrices $S_1, S_2$ (each with less rows and columns than $S$) and rows $x_1,\dots,x_{k-1}$ of $S_1$, $y_1,\dots,y_{k-1}$ of $S_2$, such that $S$ is isomorphic to $(S_1,x_1,\dots,x_{k-1}) \otimes_k (S_2,y_1,\dots,y_{k-1})$. Again, we will abuse notation and write $S = (S_1,x_1,\dots,x_{k-1}) \otimes_k (S_2,y_1,\dots,y_{k-1})$, and simply $S = S_1 \otimes_k S_2$ when the special rows are clear from the context.

The following is our first main result.

\begin{thm}
\label{thm:main1}
Let $S\in \R^{m\times n}$, and $k\geq 1$ be constant. There is an algorithm that runs in time polynomial in $n,m$ and determines whether $S$ is a $k$-\op{} and, in case it is, outputs two matrices $S_1, S_2$ and, if $k\geq 2$, special rows $x_1,\dots, x_{k-1}$ of $S_1$, $y_1,\dots, y_{k-1}$ of $S_2$, such that $S=(S_1,x_1,\dots, x_{k-1})\otimes_k (S_2, y_1,\dots, y_{k-1})$.
\end{thm}

The proof of Theorem~\ref{thm:main1} is by reduction to symmetric submodular function minimization using the concept of \emph{mutual information} from information theory.  Somewhat surprisingly, we do not know of a simpler proof of Theorem~\ref{thm:main1}. A straightforward implementation of our algorithm runs in $O(m^{k+3}(m+n))$ time.

In Section \ref{sec:def} (see Corollary \ref{cor:glued}), we show that the simplicial glued product corresponds to the $k$-\op{} of matrices if we represent a polytope via a slack matrix. As a consequence of this relationship and of Theorem~\ref{thm:main1}, we obtain the following:

\begin{thm} \label{thm:main2}
Given a constant $k \in \N$, and a slack matrix $S \in \R^{m \times n}$ of a polytope $P$, 
there is an algorithm that is polynomial in $m, n$ which correctly determines if $P$ is affinely isomorphic to a 
simplicial glued product $P_1 \times_k P_2$ and, in case it is, outputs two matrices $S_1, S_2$ such that $S_i$ is the slack matrix of $P_i$, for $i \in [2]$.
\end{thm}

A couple of remarks are in order.  First, since affine transformations do not preserve the property of being a simplicial glued product,  determining whether a polytope $P$ is affinely
isomorphic to a simplicial glued product is a different problem than that of determining whether $P$ \emph{equals} $P_1\times_k P_2$ for some polytopes $P_1, P_2$ and some $k$.  Next, since $P_1 \times_1 P_2=P_1\times P_2$, the Cartesian product is a special case of the simplicial glued product.  Therefore, as a special case of Theorem~\ref{thm:main2}, we obtain a polynomial time algorithm to test if a polytope is a Cartesian product.

Polytopes that have a $0/1$-valued slack matrix are called \emph{$2$-level polytopes}. These form a rich class of polytopes including stable set polytopes of perfect graphs, Birkhoff and Hanner polytopes, and stable matching polytopes (see~\cite{aprile2018thesis, aprile20182, macchia2018two} for more examples),  and they have been the topic of many recent investigations (see~\cite{Grande16,kupavskii2020binary}). In the second part of the paper, we use the tools developed in the first part to shed some light on this intriguing class. As a starting point, we pose the following conjecture. 
 
\begin{conj} \label{conj:2levelrec}
Given $S \in \{0,1\}^{m \times n}$, there is an algorithm that is polynomial in $m,n$ which correctly determines if $S$ is the slack matrix of a polytope.
\end{conj}
 
Given the wide variety of $2$-level polytopes, Conjecture~\ref{conj:2levelrec} appears difficult to settle. We provide some evidence 
 by proving it for two restricted classes of $2$-level polytopes. 
  By applying Theorem \ref{thm:main1}, we show that it holds for $2$-level matroid base polytopes~\cite{Grande16}\footnote{A special case of Theorem \ref{thm:main1} can be easily employed to show that Conjecture~\ref{conj:2levelrec} holds for Hanner polytopes, see~\cite{aprile2017vertices} for a definition.} and for stable set polytopes of perfect graphs. 
 
\begin{thm} \label{thm:main3}
Given $S \in \{0,1\}^{m \times n}$, there is an algorithm that is polynomial in $m,n$ which correctly determines whether $S$ is the slack matrix of a $2$-level matroid base polytope.
\end{thm}

 \begin{thm}\label{thm:recog-slack-stab(G)}
 Given $S\in \{0,1\}^{m\times n}$, there is an algorithm that is polynomial in $m,n$ which correctly determines whether $S$ is the slack matrix of the stable set polytope of a perfect graph.
 \end{thm}
 
 We will also describe a connection between simplicial glue products and clique cut-sets of graphs.
 In light of our results, one might ask how prevalent the $k$-product operation is among $2$-level polytopes. As an experimental answer, in \cite{aprile2018thesis} it is shown that, up to dimension $7$, roughly half of the slack matrices of  $2$-level polytopes are $k$-products for some $k$.

 \smallskip
 
\noindent {\bfseries{Paper Outline.}}
In Section~\ref{sec:def} we study properties of simplicial glued products, and of $k$-\op{}s of slack matrices. In Section~\ref{sec:algorithms} we give algorithms to efficiently recognize 1-\op{}s and $k$-\op{}s, as well as showing a unique decomposition result for 1-\op{}s.
In Section~\ref{sec:slackmatroids} we apply the previous results to slack matrices of matroid base polytopes, obtaining Theorem \ref{thm:main3}. In Section~\ref{sec:stab(G)k-sum} we deal with stable set polytopes of perfect graphs, and their slack matrices, proving Theorem~\ref{thm:recog-slack-stab(G)}. 

The results presented in this paper are contained in the PhD thesis of the first author \cite{aprile2018thesis}, to which we refer for further details.

 \smallskip

\noindent  {\bfseries{Conference version.}} This paper is the journal version of the extended abstract~\cite{aprilerecognizing}. Note that~\cite{aprilerecognizing} does not deal with glued products and general $k$-products, and instead mainly focuses on $1$-products and $2$-products. Results from Section~\ref{sec:slackmatroids} appear in~\cite{aprilerecognizing} without proofs. Section~\ref{sec:stab(G)k-sum} only appears in this journal version.

\section{Simplicial glued products, slack matrices, and $k$-\op{}s}\label{sec:def}

In this section, after proving geometric properties of simplicial glued products and recalling preliminary results on slack matrices,
we show that the operation of $k$-\op{} essentially preserves the property of being a slack matrix.

\begin{lem}\label{lem:gluedprod}
Let $P=P_1 \times_k P_2$ be the simplicial glued product of $P_1$ and $P_2$ with respect to $x_1,\dots,x_{k-1}$, $y_1,\dots, y_{k-1}$, where $x_1\geq 0,\dots,x_{k-1}\geq 0$ and $y_1\geq 0,\dots,y_{k-1}\geq 0$ are facets of $P_1$ (resp., $P_2$). Then:
\begin{enumerate}
    \item[(i)] $P= (P_1 \times P_2) \cap H$, where $H=\{x_i=y_i \, :i \in [k-1] \}$.
    \item[(ii)]$V(P)= V_0 \cup V_1 \cup \dots \cup V_{k-1}$,
     where we define $$V_0:=\{\{x \in V(P_1)\hbox{ such that $x_1=\dots=x_{k-1}=0$}\} \times \{ y \in V(P_2) \hbox{ such that $y_1=\dots=y_{k-1}=0$}\}\},$$ and, for $i=1,\dots, k-1$, $$V_i:=\{\{x \in V(P_1) \hbox{ such that $x_i=1$}\} \times \{y \in V(P_2) \hbox{ such that $y_i=1$}\}\}.$$   
      \item[(iii)] Fix a non-redundant, full-dimensional inequality description of $P_1$ and $P_2$.  
      Consider the description of $P$ obtained by juxtaposing the descriptions of $P_1$ and $P_2$ and replacing $y_i$ with $x_i$ for each $i=1,\dots, k-1$. Then any inequality of this description that is redundant involves variables $x_1,\dots,x_{k-1}$ only. 
      \item[(iv)] The description obtained in the previous part is non-redundant, apart from repeated inequalities $x_i\geq 0$, $i\in [k-1]$, and the inequality  $\sum_{i=1}^{k-1} x_i\leq 1$ that may possibly be repeated and/or redundant.



\end{enumerate}
\end{lem}

\begin{proof}
\begin{enumerate}
    \item[(i)] A proof of this fact can be found in~\cite[Lemma 4]{tiwary2020extension}. For completeness, we show it again here. 
    
    ``$\subseteq$'' inclusion is clear.
To prove the opposite inclusion, consider a point $p = (x^*,y^*) \in (P_1 \times P_2) \cap H \subseteq \R^{m_1} \times \R^{m_2}$, for some $x^* \in P_1$ and $y^* \in P_2$. Using a density argument, it suffices to show that $(x^*,y^*) \in P_1\times_k P_2$ when $(x^*,y^*)$ has all rational coordinates. Then $x^*$ is a convex combination of the vertices of $P_1$ and $y^*$ is a convex combination of the vertices of $P_2$, where the coefficients are all rational:
\[
x^* = \displaystyle\sum_{i = 1}^{n_1} \lambda_i v_i \quad\text{and}\quad
y^* = \displaystyle\sum_{j = 1}^{n_2}\mu_j w_j \qquad\text{with}\quad \sum_i\lambda_i=\sum_j\mu_j=1,\quad \lambda_i,\mu_j \in \Q_+\,.
\]
Then there exists a positive integer $K$ such that $K\lambda_i \in \N$ and $K\mu_j \in \N$ for every $i \in [n_1]$ and every $j \in [n_2]$. Moreover, $K = K\sum_i \lambda_i= K\sum_j \mu_j$.

Let us partition the set of vertices of $P_1$ that occur in the convex combination into $k$ subsets, according to coordinates $(x_1,\dots, x_{k-1})$. For $\ell=1,\dots,k-1$, let $V_1^\ell$ be the set of $v_i$'s with $\ell$-th coordinate equal to 1, and let $V_1^0$ be the set of the remaining $v_i$'s (i.e., with the first $k-1$ coordinates all equal to 0, since each $v_i$ has at most one of the first $k-1$ coordinates equal to 1). The sets $V_2^0,\dots,V_2^{k-1}$ are defined similarly.
For $x^*$, we have:
\[
x^* = \sum_{i: v_i \in V_1^{0}} \lambda_iv_i + \dots+ \sum_{i: v_i \in V_1^{k-1}} \lambda_iv_i\,.
\]
We split in the same way the identity $K = K \sum_i \lambda_i$. Thus $K = \alpha_{0}+\dots+ \alpha_{k-1}$, where
$\alpha_{\ell} := \sum_{i:v_i \in V_1^{\ell}} (K\lambda_i)$ for $\ell=0,\dots,k-1$. 
Applying the same reasoning to $y^*$, we get that $K = \beta_{0} + \dots + \beta_{k-1}$, where
$\beta_{\ell} := \sum_{j: w_j \in V_2^{\ell}} (K\mu_j)$ for $\ell=0,\dots,k-1$.

Since the first $n-1$ coordinates of the $v_i$'s and $w_j$'s are $0/1-$valued, we have that $Kx^*_\ell = \alpha_{\ell}$ and $Ky^*_\ell = \beta_{\ell}$ for $\ell=1,\dots,k-1$.
Exploiting the fact that the first $k-1$ coordinates of $x^*$, $y^*$ are equal, we have that $\alpha_{\ell} = \beta_{\ell}$ for $\ell=1,\dots,k-1$. These identities jointly  imply that $\alpha_{0} = \beta_{0}$.

Fix $\ell \in \{0,\dots,k-1\}$. 
The coefficients $\alpha_{\ell}$, $\beta_{\ell}$ coincide with the number of vectors $v_i$ in $V_1^{\ell}$ and $w_j$ in $V_2^{\ell}$ when counted with their multiplicity $K\lambda_i$ in the identity of $Kx^*= K\sum_i\lambda_i v_i$ and $K\mu_j$ in $Ky^*=K \sum_j\mu_j w_j$ respectively. Consider then the multiset $\overline{V}_1^{\ell}$ containing each vector $v_i$ with multiplicity $K\lambda_i$ and, similarly $\overline{V}_2^{\ell}$ containing each vector $w_j$ with multiplicity $K\mu_j$.
As $|\overline{V}_1^{\ell}|=|\overline{V}_2^{\ell}|= \alpha_{\ell}$, there exists a bijection $\Phi_{\ell}$ from the first to the latter. Let $\overline{\Phi}_\ell$ be the truncation of $\Phi_\ell$ excluding coordinates $y_1,\dots,y_{k-1}$.   Hence, the vectors  $(v_i,\overline{\Phi}_\ell(v_i))$ are vertices of $P$ for every $v_i \in \overline{V}_1^{\ell}$. 

We can now express $p = (x^*,y^*)$ as:
\[
(x^*,y^*) = \frac{1}{K} \Bigg( \sum_{v_i \in \overline{V}_1^{0}} (v_i,\overline{\Phi}_0(v_i)) +\dots +\sum_{v_i \in \overline{V}_1^{k-1}} (v_i,\overline{\Phi}_{k-1}(v_i))\Bigg)
\]
This shows that $(x^*,y^*)$ lies in the convex hull of vertices of $P$.

 \smallskip
 
 \item[(ii)] $V\subseteq V_0 \cup V_1 \cup \dots \cup V_{k-1}$ holds since points in $V_0 \cup \dots \cup V_{k-1}$ are vertices of $P_1 \times P_2$, and by part (i), $P\subseteq P_1\times P_2$. For the opposite inclusion, take any $v \in V_i$ for some $i =0,\dots, k-1$. By definition of glued product, $v\in P$. We now show that $v\in V(P)$. The components of $v$ in coordinates $x_1,\dots,x_{k-1}$ are $0/1$. Since no $0/1$ point can be written as a convex combination of other $0/1$ points, all points in a convex combination giving $v$ must agree with $v$ on coordinates corresponding to $x_1,\dots, x_{k-1}$. Hence, if $v \notin V(P)$, $v$ is obtained as a convex combination of points from $V_i$. This implies that the restriction of $v$ to the coordinates of $P_1$ is not a vertex of $P_1$, a contradiction.
 

 
 
 \smallskip

 \item[(iii)] Let $P_1=\{x : A^1 x \leq b_1 \}$, $P_2=\{y : A^2 y \leq b_2\}$ be irredundant linear descriptions of $P_1$ and $P_2$ as required by the lemma. Notice that the two descriptions do not contain implicit equations. Let $Ax + By \leq b$ the system obtained as described in the lemma. We remark that such system indeed describes $P$, thanks to part i) of the lemma.
 
 Consider an inequality from $A x +By \leq b$ that involves variables other than $x_1,\dots, x_{k-1}$, and suppose it is redundant. Without loss of generality, we assume this is an (irredundant) inequality of $A^1x\leq b_1$, hence we denote it by $ax \leq \beta$. $ax\leq \beta$ can be written as a conic combination of other inequalities from $Ax + By \leq b$. Hence in particular $a=a^1+a^2$, where $a^1$ (resp. $a^2$) is a conic combination of rows of $A^1$ (resp. $A^2$), and there are $\beta_1, \beta_2$ with $\beta_1+\beta_2=\beta$ such that $a^1x\leq \beta_1$ (resp. $a^2y\leq \beta_2$)  is valid for $P_1$ (resp. $P_2$). Now, $a^2$ must have entry 0 in correspondence of any $y_j$ variable with $j>k-1$. Hence, $a^2y\leq \beta_2$ is a valid inequality of $P_2$ involving variables $y_1,\dots, y_{k-1}$ only, implying that $a^1x\leq \beta_1$ involves variables other than $x_1,\dots, x_{k-1}$. Moreover, since by definition of simplicial glued product the projection of $P_2$ over $y_1,\dots, y_{k-1}$ is the same as the projection of $P_1$ over $x_1,\dots, x_{k-1}$, we have that $a^2x\leq \beta_2$ is valid for $P_1$. But then $ax \leq \beta$, which is facet-defining for $A^1x\leq b_1$, becomes redundant after adding $a^2x\leq \beta_2$ to the system: this implies that $a^2x\leq \beta_2$ is obtained by scaling $ax \leq \beta$ by a positive factor, but this contradicts the fact that the latter inequality involves variables other than $x_1,\dots,x_{k-1}$, and the former does not.
 
  \item[(iv)] Consider the description $P$ as in the previous part. In order to conclude the thesis, we must show that no inequality that involve variables $x_1,\dots, x_{k-1}$ only and is different from 
  $\sum_{i=1}^{k-1} x_i\leq 1$ is redundant (apart from repeated inequalities). We remark that, since the projection of $P_1$ on variables $x_1,\dots, x_{k-1}$ is a $k-1$ dimensional simplex, the only  inequalities involving variables $x_1,\dots, x_{k-1}$ only that can be non-redundant in $P_1$ are $x_i\geq 0$, $i\in [k-1]$, and $\sum_{i=1}^{k-1} x_i\leq 1$. The same holds for $P_2$. Hence we only need to show that inequalities $x_i\geq 0$, $i\in [k-1]$ are non-redundant in our description of $P$. By contradiction, assume that $x_i\geq 0$ is redundant: hence $x_i=0$ defines the affine hull of a face of $P$ that is strictly contained in a facet of $P$. In particular, there is an inequality of $P$ (different from $x_i\geq 0$) satisfied at equality by all vertices $(x,y)\in V(P)$ with $x_i=0$. This inequality is also facet defining in $P_1$ (without loss of generality). But this implies that $x_i\geq 0$ is redundant in the description of $P_1$, a contradiction with the definition of simplicial glued product.

 \end{enumerate}
\end{proof}

We will denote the set of column vectors of a matrix $S$ by $\col(S)$. The following characterization of slack matrices is due to~\cite{gouveia2013nonnegative}.

\begin{thm}[Gouveia~\emph{et al.} \cite{gouveia2013nonnegative}]\label{thm:char_slack-matrices}
Let $S\in \R^{m \times n}$ be a nonnegative matrix of rank at least 2. Then $S$ is the slack matrix of a polytope if and only if $\conv(\col(S)) = \aff(\col(S)) \cap \R_+^m$. 
\end{thm}

Throughout the paper, we will assume that the matrices we deal with are of rank at least $2$, so we may apply Theorem \ref{thm:char_slack-matrices} directly. 
We also recall the following useful fact:

\begin{lem}[Gouveia \emph{et al.}~\cite{gouveia2013nonnegative}]\label{lem:slackconvcol}
If $S$ is the slack matrix of a polytope $P$ with $\dim(P) \geq 1$, then $P$ is affinely isomorphic to $\conv(\col(S))$. In addition, we have $\dim(P)=\rk(S)-1$.
\end{lem}

We point out that the slack matrix of a polytope $P$ is not unique, as it depends on the given descriptions of $P$. We say that a slack matrix is \emph{non-redundant} if its rows bijectively correspond to the facets of $P$.
In particular, non-redundant slack matrices do not contain two identical rows or columns, nor rows or columns which are all zeros, or all non-zeros. 
A non-redundant slack matrix associated to a polytope is unique up to permuting rows and columns, and scaling rows by strictly positive reals.   

Given a slack matrix $S$, a canonical way to construct a polytope $P$ whose slack matrix is $S$ is given by the \emph{slack embedding}, defined as $P:=\conv(\col(S))$, see Lemma~\ref{lem:slackconvcol}. Moreover, if $S$ is a slack matrix of polytopes $P_1,P_2$, then $P_1$ and $P_2$ are affinely isomorphic. For proofs of these facts and more properties of slack matrices, see, e.g.,~\cite{bohn2019enumeration}.

We remark that, from the algorithmic point of view, one can assume to deal with non-redundant slack matrices. 
Indeed, in light of Lemma \ref{lem:slackconvcol}, one can efficiently check whether some columns are redundant (i.e., are contained in the convex hull of the others) by using linear programming. On the other hand, zero rows are always redundant, and whether a non-zero row is redundant can be checked by verifying that its set of zeros (i.e., the vertices lying on the corresponding face) is maximal.

  We now prove the main lemma of the section. We first write the special case of $1$-\op{}s, which shows that a 1-product is a slack matrix if and only if the factors are. We refer to \cite{aprile2018thesis, aprilerecognizing} for the proof, which is an easy application of Theorem \ref{thm:char_slack-matrices}. 
  
  \begin{lem}\label{lem:1-sum_slack}
Let $S\in\R_+^{m\times n}$ and let $S_i\in\R_+^{m_i\times n_i}$ for $i \in [2]$ such that $S=S_1\otimes S_2$. $S$ is the slack matrix of a polytope $P$ if and only if there exist polytopes $P_i$, $i \in [2]$ such that $S_i$ is the slack matrix of $P_i$ and $P$ is affinely isomorphic to $P_1 \times P_2$.
\end{lem}

\begin{lem}\label{lem:k-sum_slack}
Let $k\geq 2$, $S\in\R_+^{m\times n}$ and let $S_i\in\R_+^{m_i\times n_i}$ for $i=1,2$ such that $S=(S_1,x_1,\dots,x_{k-1}) \otimes_k (S_2,y_1,\dots,y_{k-1})$ for some special rows $(x_1,\dots,x_{k-1})$ of $S_1$, and $(y_1,\dots,y_{k-1})$ of $S_2$. 
\begin{enumerate}
\item If $S_1, S_2$ are slack matrices, then $S$ is a slack matrix.
\item If $S$ is a slack matrix, let $S'_1=S_1 +(\mathbf{1}-x_1-\dots -x_{k-1})$ and construct $S_2'$ similarly\footnote{Here and throughout the paper the $+$ operation takes as input a matrix $M$ and a row vector $r$ of the same size of $\col(M)$, and outputs the matrix $\binom{M}{r}$.}. Then $S_1', S_2'$ are slack matrices.
\end{enumerate}
\end{lem}
\begin{proof}
(i). Let $P_i := \conv(\col(S_i)) \subseteq \R^{m_i}$ for $i = 1,2$. Without loss of generality, $x_1, \dots, x_{k-1}$  can be assumed to be the first $k-1$ rows of $S_1$, and similarly for $y_1,\dots,y_{k-1}$ and $S_2$. Hence, for a point $x \in \R^{m_1}$, we overload notation and denote by $x_i$ the $i$-th coordinate of $x$, and similarly for $y\in \R^{m_2}$. Notice that, by the definition of $k$-product, $\conv(\col(S))$ is the simplicial glued product of $P_1$, $P_2$ with respect to $x_1, \dots, x_{k-1}$, $y_1,\dots,y_{k-1}$. 

We remark that $S$ is a submatrix of a slack matrix of $(P_1\times P_2)\cap H$, where $H$ is the hyperplane defined by the equations $x_1=y_1,\dots,  x_{k-1} = y_{k-1}$. Indeed, using Lemma~\ref{lem:1-sum_slack}, the slack matrix of $P_1 \times P_2$ is given by $S_1 \otimes S_2$. Adding constraints $x_i=y_i$ for $i \in [k-1]$ implies that columns of $P_1 \times P_2$ whose entries $x_i$ and $y_i$ differ are not vertices of $(P_1\times P_2)\cap H$, and can therefore be excluded. After such exclusions, rows $x_i$ and $y_i$ are copies of one another, and in a non-redundant slack matrix we can exclude one of them. We deduce therefore that there exists a slack matrix $S'$ of $(P_1\times P_2)\cap H$ which has all and only the rows of $S$, and possibly more columns.
However, by Lemma \ref{lem:gluedprod} part i) we have that $\conv(\col(S))=P_1\times_k P_2= (P_1\times P_2)\cap H$, implying that $S'$ does not have any more column than $S$, i.e. $S=S'$ is a slack matrix. 

\medskip

\noindent (ii). Let $S = (S_1,x_1,\dots,x_{k-1})\otimes_k (S_2,y_1,\dots,y_{k-1})$ be a slack matrix. We show that $S_1'=S_1+(\mathbf{1}-x_1-\dots-x_{k-1})$ is a slack matrix (where we let ($\mathbf{1}-x_1-\dots-x_{k-1})$ be the $k$-th row of $S'_1$), the argument for $S_2'$ being exactly the same. By Theorem \ref{thm:char_slack-matrices}, we have $\aff(\col(S))\cap\R^m_+=\conv(\col(S))$, and we will show that the same holds for $S_1'$. We use a similar notation as in the proof of Lemma \ref{lem:gluedprod}: we assume that $x_1,\dots,x_{k-1}, \mathbf{1}-x_1-\dots-x_{k-1}$ are the first $k$ rows of $S_1'$, and for $\ell=1,\dots,k$ we denote by $V_1^\ell$ the set of columns of $S_1'$ with $\ell$-th coordinate equal to $1$.

Let $x^*\in \aff(\col(S_1'))\cap \R_+^{m_1}$, one has $x^*=\sum_i \lambda_i v_i= \sum_{v_i \in V_1^1} \lambda_iv_i + \dots+ \sum_{v_i \in V_1^{k}} \lambda_iv_i$, with $\sum_i \lambda_i=1$. In particular, $x^*_\ell=\sum_{v_i \in V_1^\ell} \lambda_i\geq 0$ for $\ell=1,\dots,k$ with $x_1^*+\dots+x_k^*=1$. 
We now extend $x^*$ to a point $\tilde{x}\in \aff(\col(S))$ as follows: for $\ell=1,\dots, k$, fix a column $u_\ell$ of $S_2$ with the $\ell$-th coordinate equal to 1 if $\ell<k$, and with the first $k-1$ coordinates all equal to 0 if $\ell=k$ (such columns exist since by assumption every $S_2[0],\dots,S_2[k-1]$ is non-empty). Then, for $\ell=1,\dots, k$, map each $v_i \in V_1^\ell$ to the column of $S$ consisting of $v_i$ (without its $k$-th coordinate) followed by $u_\ell$ (without the coordinates corresponding to $y_1,\dots,y_{k-1}$). We denote such column by $w_i$, for $i=1,\dots,n_1$, and let $\tilde{x}=\sum_i \lambda_i w_i$. Now, we claim that $\tilde{x} \in\R^m_+$: indeed, by construction of the $w_i$'s, every component of $\tilde{x}$ is equal to a component of $x^*$, or to a sum of a (possibly empty) subset of $\{x_1^*,\dots, x^*_k\}$, according to the corresponding component of the $u_\ell$'s.
Hence we have that  $\tilde{x} \in  \aff(\col(S))\cap\R^m_+=\conv(\col(S))$. We claim that this implies that $x^*\in  \conv(\col(S_1'))$: indeed, if $\tilde{x}=\sum_i \mu_i w_i$, with $\mu_i\geq 0$ for $i=1,\dots,n_1$, and $\sum_i \mu_i=1$ then it follows that $x^*=\sum_i \mu_i v_i$. This is trivial except for the $k$-th coordinate, which is equal to $\sum_{i : v_i \in V_1^k} \mu_i=1-\sum_{i:v_i \in V_1^1} \mu_i-\dots-\sum_{i:v_i \in V_1^{k-1}} \mu_i=1-x_1^*-\dots-x_{k-1}^*=x_k^*$, where the latter equality holds by construction. Hence we conclude that $S_1'$ is a slack matrix.
\end{proof}

We now give an explicit relation between the $k$-product and the simplicial glued product.

\begin{cor}\label{cor:glued}Let $k \in \N$, $P$ be a polytope, and $S$ a non-redundant slack matrix of $P$.  
 Then $P$ is affinely isomorphic to the simplicial glued product of polytopes $P_1,P_2$ with respect to variables $x_1,\dots, x_{k-1}$ of $P_1$, $y_1,\dots, y_{k-1}$ of $P_2$ if and only if
 $S=S_1\otimes_k S_2$ for some matrices $S_1,S_2$. Moreover, if the latter happens, $S_1'$ (resp. $S_2'$) obtained from $S_1$ (resp. $S_2$) as in Lemma \ref{lem:k-sum_slack} is a slack matrix of $P_1$ (resp. $P_2$), and the special rows of $S_1$ (resp.~$S_2$) correspond to inequalities $x_i\geq 0$ (resp. $y_i\geq 0$) for $i\in [k-1]$.
\end{cor}

\begin{proof}
Assume first that $S$ is a $k$-product of $S_1$, $S_2$. $S_1',S_2'$ as described in Lemma~\ref{lem:k-sum_slack} are slack matrices of certain polytopes $P_1,P_2$. 
Then thanks to Lemma \ref{lem:slackconvcol}, $P$ is affinely isomorphic to $\conv(\col(S))$ which, arguing as in the proof of Lemma \ref{lem:k-sum_slack}, part (i), is affinely isomorphic to the glued product of $P_1$, $P_2$. This shows the ``if'' direction. 

Conversely, assume that $P$ is affinely isomorphic to the simplicial glued product $P_1\times_k P_2$ with respect to variables $x_1,\dots, x_{k-1}$ of $P_1$, $y_1,\dots, y_{k-1}$ of $P_2$. Fix non-redundant, full-dimensional descriptions of $P_1$, $P_2$  that contain inequalities $x_i\geq 0$ (resp. $y_i\geq 0$)  for $i\in [k-1]$, and (if facet-defining) the inequality $\sum_{i=1}^{k-1} x_i\leq 1$ (resp. $\sum_{i=1}^{k-1} y_i\leq 1$).
Now, consider the description of $P$ obtained from such descriptions as in part (iv) of Lemma \ref{lem:gluedprod}. 
Let $S'$ be the corresponding non-redundant slack matrix of $P$. 
By construction $S'$ is a $k$-product of two matrices $\tilde S_1$, $\tilde S_2$, with special rows corresponding to $x_i\geq 0$ for $i\in [k-1]$, where the columns of $\tilde S_1$ (resp. $\tilde S_2$) correspond to vertices of $P_1$ (resp. $P_2$), the rows of $\tilde S_1$ (resp. $\tilde S_2$) correspond to inequalities of the fixed description of $P_1$ (resp. $P_2$), and all inequalities of the latter description are injectively mapped to rows of $\tilde S_1$ (resp.~$\tilde S_2$), with  possibly the exception of $\sum_{i=1}^{k-1} x_i\leq 1$ (resp. $\sum_{i=1}^{k-1} y_i\leq 1$). Adding back the row corresponding to such inequality gives $\tilde{S}_1'$ (resp. $\tilde{S}_2'$), that is a slack matrix of $P_1$ (resp. $P_2$).

  Now, $S$ and $S'$ are both non-redundant slack matrices of $P$, hence they are obtained from one another by scaling rows by positive factors and permuting rows and columns. Hence $S$ is a $k$-product of matrices $S_1, S_2$ that are obtained from $\tilde S_1, \tilde S_2$ by scaling rows by positive factors. This implies that matrices $S_1'$, $S_2'$ as in the thesis are obtained from $\tilde S_1'$, $\tilde S_2'$ by scaling rows by positive factors, hence they are slack matrices of $P_1$, $P_2$ respectively.\end{proof}

The following observation justifies the idea of decomposing slack matrices via $k$-\op{s}, proving that a $k$-product of two slack matrices has strictly larger rank than theirs. The only exception is when one of the two factors is (a scaling of) the $k\times k$ identity matrix, which we denote by $I_k$ (which is the non-redundant slack matrix of a $(k-1)$-dimensional simplex): one can check that $I_k$ acts as a sort of neutral element of the $k$-product of slack matrices (see \cite{aprile2018thesis} for further details).

\begin{obs}\label{obs:k-sum}
Let $k\geq 2$ and let $P, P_1, P_2$ be such that $P$ is the simplicial glued product of $P_1, P_2$. Assume that neither $P_1$ nor $P_2$ is not a $(k-1)$-dimensional simplex. Then $\dim(P)>\max\{\dim(P_1),\dim(P_2)\}\geq k-1$.
\end{obs}
\begin{proof} Let $S$ be a slack matrix of $P$. Thanks to Corollary \ref{cor:glued}, we have that  $S$ (up to scaling rows by positive factors) is isomorphic to a $k$-product of two matrices that, after adding a row as in Lemma \ref{lem:k-sum_slack}, are slack matrices of $P_1$, $P_2$ respectively. Let $r_1,\dots, r_{k-1}$ be the special rows of $S$, and let $r=\mathbf{1}-r_1-\dots-r_k$. Using Theorem~\ref{thm:char_slack-matrices}, by possibly adding row $r$ twice to $S$ we obtain a slack matrix $S'$ of $P$ that is the $k$-product of matrices $S_1$, $S_2$ that, by Corollary~\ref{cor:glued} are slack matrices of $P_1, P_2$ respectively. Notice that $\rk(S)=\rk(S')=\dim(P)+1$ by Lemma \ref{lem:slackconvcol}.
We show that $\rk(S)>\max\{\rk(S_1),\rk(S_2)\}\geq k$, which concludes the proof thanks to Lemma \ref{lem:slackconvcol}.

As a first remark, we have that $\rk(S_1)\geq k$: indeed, choose one column $c_j$ from each $S_1[j]$ for $j=0,\dots,k-1$. The columns of $S_1$ corresponding to $c_0, \dots, c_{k-1}$ are linearly independent. Similarly we have $\rk(S_2)\geq k$. 


We now show that $\rk(S)>\rk(S_1)$, the proof for $S_2$ being the same. This will complete the proof. First, notice that $\rk(S)\geq \rk(S_1)$ since $S_1$ is a submatrix of $S'$. Now, assume by contradiction that equality holds, hence there are $t=\rk(S_1)=\rk(S)$ columns of $S_1$ that form a basis for the column space of $S_1$, and $t$ corresponding columns of $S$ that form a basis $B$ for the column space of $S$. Every column of $S$ can be written in a unique way as linear combination of columns in $B$, implying that no two columns of $S$ are identical when restricted to rows of $S_1$, but different otherwise. Hence $S_2[j]$ consists of one column only for any $j\in \{0,\dots,k-1\}$, in particular $S_2$ has exactly $k$ columns. Hence, the slack embedding of $S_2$, hence $P_2$ as well, is  a $(k-1)$-dimensional simplex, a contradiction.
\end{proof}

\section{Algorithms}\label{sec:algorithms}

In this section we study the problem of recognizing $k$-\op{}s. We first focus on the following problem: given a matrix $S$, we want to determine whether $S$ is a $1$-\op{}, and find matrices $S_1, S_2$ such that $S=S_1\otimes S_2$. Since we allow the rows and columns of $S$ to be permuted in an arbitrary way, the problem is non-trivial. 
In the second part of the section, we extend our methods to the problem of recognizing $k$-\op{}s.

We begin with a preliminary observation, which is the starting point of our approach. Suppose that a matrix $S$ is a 1-\op{} $S_1\otimes S_2$. Then the rows of $S$ can be partitioned into two sets $R_1, R_2$, corresponding to the rows of $S_1, S_2$ respectively. We write that $S$ is a 1-\op{} \emph{with respect to} the partition $R_1,R_2$. A column of the form $(a_1,a_2)$, where $a_i$ is a column vector with components indexed by $R_i$ ($i \in [2]$), is a column of $S$ if and only if $a_i$ is a column of $S_i$ for each $i \in [2]$. Moreover, the number of occurrences of $(a_1,a_2)$ in $S$ is just the product of the number of occurrences of $a_i$ in $S_i$ for $i \in [2]$. Under uniform probability distributions on the columns of $S$, $S_1$ and $S_2$, the probability of picking $(a_1,a_2)$ in $S$ is the product of the probability of picking $a_1$ in $S_1$ and that of picking $a_2$ in $S_2$. We will exploit this intuition below.

\subsection{Recognizing 1-\op{}s via submodular minimization} 

First, we recall some notions from information theory, see~\cite{cover2012elements} for a more complete exposition. 
%
Let $A$ and $B$ be two discrete random variables with ranges $\mathcal{A}$ and $\mathcal{B}$ respectively. The \emph{mutual information} of $A$ and $B$ is:
$$
I(A;B) = \sum_{a\in \mathcal{A}, b\in \mathcal{B}} \pr(A=a,B=b) \cdot \log_2 \frac{\pr(A=a,B=b)}{\pr(A=a)\cdot \pr(B=b)}. 
$$
The mutual information of two random variables measures how close is their joint distribution to the product of the two corresponding marginal distributions. 
%
%

We will use the following facts, whose proof can be found in \cite{cover2012elements, krause2005near}. Let $C_1, \ldots, C_m$ be discrete random variables. For $X \subseteq [m]$ we consider the random vectors $C_{X} := (C_i)_{i \in X}$ and $C_{\overline{X}} := (C_i)_{i \in \overline{X}}$, where $\overline{X} := [m] \setminus X$. The function $f : 2^{[m]} \to \R$ such that
\begin{equation}
\label{eq:def_f}
f(X) := I(C_X;C_{\overline{X}})
\end{equation}
will play a crucial role.

\begin{prop}\label{prop:mutual_info}
\begin{enumerate}[(i)]
\item\label{prop:indep} For all discrete random variables $A$ and $B$, we have $I(A;B)\geq 0$, with equality if and only if $A$ and $B$ are independent.
\item\label{prop:submodular}
If $C_1, \ldots, C_m$ are discrete random variables, then the function $f$ as in~\eqref{eq:def_f} is submodular.
\end{enumerate}
\end{prop}

Let $S$ be an $m \times n$ matrix. Let $C := (C_1,\ldots,C_m)$ be a uniformly chosen random column of $S$. That is, $\pr(C = c) = \mu(c)/n$, where $\mu(c)$ denotes the number of occurrences in $S$ of the column $c \in \col(S)$. 

Let $f: 2^{[m]} \rightarrow \R$ be defined as in~\eqref{eq:def_f}. We remark that the definition of $f$ depends on $S$, which we consider fixed throughout the section. The set function $f$ is non-negative (by Proposition~\ref{prop:mutual_info}.\ref{prop:indep}), symmetric (that is, $f(X)=f(\overline{X})$) and submodular (by Proposition~\ref{prop:mutual_info}.\ref{prop:submodular}). 

The next lemma shows that we can determine whether $S$ is a 1-\op{} by minimizing $f$. Its proof can be found in \cite{aprile2018thesis,aprilerecognizing}.

\begin{lem}\label{lem:mutual_info}
Let $S\in \R^{m\times n}$, and $\emptyset\neq X\subsetneq [m]$. Then $S$ is a $1$-\op{} with respect to $X,\overline{X}$ if and only if $C_X$ and $C_{\overline{X}}$ are independent random variables, or equivalently (by Proposition \ref{prop:mutual_info}.\ref{prop:indep}), if and only if $f(X) = 0$.
\end{lem}
Notice that one can efficiently reconstruct $S_1$, $S_2$ once we identified $X$ such that $f(X)=0$. In particular, if the columns of $S$ are all distinct, then $S_1$ consists of all the distinct columns of $S$ restricted to the rows of $X$, each taken once, and $S_2$ is obtained analogously from $S$ restricted to the rows of $\overline{X}$. The last ingredient we need is that every (symmetric) submodular function can be minimized in polynomial time. Here we assume that we are given a polynomial time oracle to compute our function.

\begin{thm}[Queyranne~\cite{queyranne1998minimizing}] \label{thm:submodular}
There is a polynomial time algorithm that outputs a set $X$ such that $X \neq \emptyset, [m]$ and $f(X)$ is minimum, where $f : 2^{[m]}\rightarrow \R$ is any given symmetric submodular function.
\end{thm}

We have now all the ingredients to conclude the following:

 \begin{thm}\label{thm:1-prod_recog}
 Let $S\in \R^{m\times n}$. There is an algorithm that is polynomial in $m,n$ and determines whether $S$ is a 1-\op{} and, in case it is, outputs two matrices $S_1, S_2$ such that $S=S_1\otimes S_2$.
 \end{thm}
\begin{proof} It is clear that $f(X)$ can be computed in polynomial time for any $X$. It suffices then to run Queyranne's algorithm to find $X$ minimizing $f$. If $f(X)>0$, then $S$ is not a 1-\op. Otherwise, $f(X)=0$ and $S_1$, $S_2$ can be reconstructed as described in the proof of Lemma~\ref{lem:mutual_info}.
\end{proof}


We conclude the section with a decomposition result which will be useful in the next section. We call a matrix \emph{irreducible} if it is not a 1-\op. The result below generalizes the fact that a polytope can be uniquely decomposed as a cartesian product of ``irreducible'' polytopes. 


\begin{lem}\label{lem:unique_decomposition}
Let $S\in \R^{m\times n}$ be a 1-\op. Then there exists a partition $\{X_1,\dots, X_t\}$ of $[m]$ such that:
\begin{enumerate}
\item $S$ is a $1$-\op{} with respect to $X_i, \overline{X_i}$ for all $i \in [t]$;
\item for all $i \in [t]$ and all proper subsets $X$ of $X_i$, $S$ is not a 1-\op{} with respect to $X,\overline{X}$;
\item the partition $X_1, \dots, X_t$ is unique up to permuting the labels.
\end{enumerate}
In particular, if $S$ has all distinct columns, then there are matrices $S_1,\dots, S_t$ such that $S=S_1\otimes\dots\otimes S_t$, each $S_i$ is irreducible, and the choice of the $S_i$'s is unique up to renaming and permuting columns.
\end{lem}
\begin{proof}
Let $f:2^{[m]} \to \R$ be the function defined in Equation~\ref{eq:def_f}. Let $\mathcal M =\{X \subseteq [m] \mid f(X)=0\}$.  Let $X_1, \dots, X_t$ be the minimal (under inclusion) non-empty members of $\mathcal M$.  Since $f$ is non-negative and submodular, if $f(A)=f(B)=0$, then $f(A\cap B)=f(A\cup B)=0$. By minimality, this implies that $X_i \cap X_j = \emptyset$ for all $i \neq j$.  Since $f$ is symmetric, $\bigcup_{i \in [t]} X_i=[m]$. By Lemma~\ref{lem:mutual_info},  $t \geq 2$ and $X_1, \dots, X_t$ satisfy (i) and (ii).  Conversely, by Lemma~\ref{lem:mutual_info}, if $\{Y_1, \dots, Y_s\}$ is a partition of $[m]$ satisfying (i) and (ii), then $Y_1, \dots, Y_s$ are the minimal non-empty members of $\mathcal M$, which proves uniqueness.  



To conclude, assume that $S$ has all distinct columns. Then as argued above each $S_i$ is obtained by picking each distinct column of $S$ restricted to rows of $X_i$ exactly once, and it is thus unique up to permutations, once $X_i$ is fixed. Each $S_i$ is irreducible thanks to the minimality of $X_i$ and to Lemma \ref{lem:mutual_info}. The fact that the $X_i$'s are unique up to renaming concludes the proof.
\end{proof}

\subsection{Extension to $k$-products}
We now extend the previous results in order to prove Theorem \ref{thm:main1} and \ref{thm:main2}. 
\begin{proof}[Proof of Theorem \ref{thm:main1}]
Recall that, if a matrix $S$ is a $k$-\op{}, then it has $k-1$ special rows that divide $S$ in submatrices $S[0], \dots, S[k-1]$, all of which are 1-\op{}s \emph{with respect to the same partition}. Hence, our algorithm starts by guessing the $k-1$ special rows, and obtaining the corresponding submatrices $S[0], \dots, S[k-1]$. Notice that, while our input matrix $S$ is a matrix with real entries, the special rows are chosen among the 0/1 rows of $S$. Let $f_0,f_1,\dots,f_{k-1}$ denote the functions $f$ as defined in \eqref{eq:def_f} with respect to the matrices $S[0], \dots, S[k-1]$ respectively, and let $\tilde{f}=\sum_{i=0}^{k-1} f_i$. Notice that $\tilde{f}$ is submodular, and is zero if and only if each $f_i$ is. Let $X$ be a proper subset of the non-special rows of $S$ (which are the rows of any of $S[0], \dots, S[k-1]$). It is an easy consequence of Lemma \ref{lem:mutual_info} that $S[0], \dots, S[k-1]$ are 1-\op{}s with respect to $X$ if and only if $\tilde{f}(X)=0$. Then $S$ is a $k$-\op{} with respect to the chosen special rows if and only if the minimum of $\tilde{f}$ is zero. If this does not happen, we proceed to the next guessing of special rows. There are $O(m^k)$ such guessings, hence polynomially many in the input size (recall that $k$ is assumed to be constant).

Alternatively, one could first repeatedly decompose each $S[j]$ and obtain a minimal partition of its row set, as described in Lemma \ref{lem:unique_decomposition}, and then check whether these $k$ partitions are refinements of a single partition $X,\bar{X}$ of the row set. This is a simple combinatorial problem that can be solved efficiently.


Once a feasible partition is found, $S_1, S_2$ can be reconstructed by first reconstructing all $S_1[j]$'s, $S_2[j]$'s and then concatenating them and adding the special rows.
\end{proof}

\begin{proof}[Proof of Theorem \ref{thm:main2}] 
In polynomial time, it is possible to remove redundant rows from the slack matrix $S$ of $P$. Hence, from now on, we assume that $S$ is non-redundant. 

By Corollary~\ref{cor:glued}, $P$ is affinely isomorphic to $P_1 \times_k P_2$, for some polytopes $P_1,P_2$ if and only if the algorithm from Theorem~\ref{thm:main1} outputs matrices $S_1, S_2$ such that $S=S_1 \otimes_k S_2$. In this case, $S_1',S_2'$ constructed as in Corollary~\ref{cor:glued} are slack matrices of $P_1, P_2$. \end{proof}



\section{Recognizing 2-level matroid base polytopes}\label{sec:slackmatroids}

\subsection{Basic definitions}

In this section, we use the results in Section \ref{sec:algorithms} to derive a polynomial time algorithm to recognize the slack matrix of a 2-level base matroid polytope. 

We start with some basic definitions and facts about matroids, and we refer the reader to \cite{oxley2006matroid} for missing definitions and details. We regard a matroid $M$ as a pair $(E,\calB)$, where $E$ is the ground set of $M$, and $\calB$ is its set of bases. The dual matroid of $M$, denoted by $M^*$, is the matroid on the same ground set whose bases are the complements of the bases of $M$. 
An element $p\in E$ is called a \emph{loop} (respectively \emph{coloop}) of $M$ if it appears in none (all) of the bases of $M$.
If $e$ is a coloop, then the \emph{deletion} of $e$ is the matroid $M \setminus e$ on $E\setminus \{e\}$ whose set of bases is $\{B \setminus \{e\} \mid B \in \mathcal B\}$; otherwise $M\setminus e$ is the matroid on $E\setminus \{e\}$ whose bases are the bases of $M$ that do not contain $e$.  The \emph{contraction} of $e$ is the matroid $M/e:=(M^* \setminus e)^*$. 
A matroid $M = (E,\calB)$ is \emph{uniform} if $\calB=\binom{E}{k}$, where $k$ is the rank of $M$. We denote the uniform matroid with $n$ elements and rank $k$ by $U_{n,k}$.

Consider matroids $M_1=(E_1,\calB_1)$ and $M_2=(E_2,\calB_2)$, with non-empty ground sets. If $E_1 \cap E_2 =\emptyset$, the \emph{1-sum} $M_1 \oplus M_2$ is defined as the matroid with ground set $E_1 \cup E_2$ and base set $\{B_1 \cup B_2 \mid B_1 \in \calB_1, B_2 \in \calB_2\}$. If, instead, $E_1 \cap E_2 = \{p\}$, where $p$ is neither a loop nor a coloop in $M_1$ or $M_2$, 
we let the \emph{2-sum} $M_1 \oplus_2 M_2$ be the matroid with ground set $(E_1 \cup E_2) \setminus \{p\}$, and base set $\{(B_1 \cup B_2) \setminus \{p\} \mid B_i\in \calB_i \mbox{ for } i \in [2] \mbox{ and } p\in B_1 \triangle B_2\}$.  A matroid $M$ is \emph{connected} if it cannot be written as the 1-sum of two matroids, each with fewer elements then $M$. It is well-known that $M_1 \oplus_2 M_2$ is connected if and only if $M_1$ and $M_2$ are both connected. We also recall the following (see~\cite[Chapters 4.2 and Chapter 8.3]{oxley2006matroid}).

\begin{prop}\label{prop:sumrank} 
Let $M$ be a matroid on ground set $E$ and let $E_1, E_2\subseteq E$ be disjoint.
\begin{enumerate}
    \item $M=M_1\oplus M_2$ for some matroids $M_1, M_2$ on ground sets $E_1$, $E_2$ if and only if $\rk(E_1)+\rk(E_2)=\rk(E)$.
    \item Let $M$ be connected. Then $M=M_1\oplus_2 M_2$ for some matroids $M_1, M_2$ on ground sets $E_1\cup\{p\}$, $E_2\cup\{p\}$ respectively if and only if there is a partition $(E_1, E_2)$ of $E$ with $\rk(E_1)+\rk(E_2)=\rk(E)+1$.
\end{enumerate}
\end{prop}

The \emph{base polytope} $B(M)$ of a matroid $M$ is the convex hull of the characteristic vectors of its bases. It is well-known that:
\[
B(M)=\{x\in\R^E_+: x(U)\leq \rk(U) \ \forall U\subseteq E, x(E)=\rk(E)\},
\]
where $\rk$ denotes the rank function of $M$.

We remark that, for any matroid $M$, the base polytopes $B(M)$ and $B(M^*)$ are affinely isomorphic via the transformation $f(x)=1-x$ and hence have the same slack matrix. 

\subsection{A high-level view of our algorithm}

Our algorithm is based on the following decomposition result, that characterizes those matroids $M$ such that $B(M)$ is 2-level (equivalently, such that $B(M)$ admits a 0/1 slack matrix). 

\begin{thm}[\cite{Grande16}] \label{thm:matroid-2-level} 
	The base polytope of a matroid $M$ is $2$-level if and only if $M$ can be obtained from uniform matroids through a sequence of 1-sums and 2-sums.
\end{thm}

The general idea is to use the algorithm from Theorem~\ref{thm:main1} to decompose our candidate slack matrix as 1-\op{}s and 2-\op{}s, until each factor corresponds to the slack matrix of a uniform matroid, which can be easily recognized. We discuss each of those steps below.

First, observe that, if $M=M_1\oplus M_2$, then $B(M)$ is the Cartesian product $B(M_1)\times B(M_2)$. Hence its slack matrix is a 1-\op{}, by Lemma \ref{lem:1-sum_slack}. Lemma \ref{lem:1-sum_matroid} shows that the converse holds. 

Second, if $M=M_1\oplus_2 M_2$, then a less trivial polyhedral relation holds: $B(M)$ is the glued product of appropriate affine transformations of $B(M_1)$, $B(M_2)$, providing a connection with the 2-\op{} of slack matrices. We will explain this connection in Lemma~\ref{obs:2sumpolydescription}.

Third, let us explain how to recognize the slack the matrix of base polytopes of uniform matroids. The base polytope of the uniform matroid $U_{n,k}$ is the $(n,k)$-hypersimplex
\(
B(U_{n,k}) = \{ x \in [0,1]^E \mid x(E) = k \}
\).
If $2\leq k \leq n-2$, the (irredundant, 0/1) slack matrix $S$ of $B(U_{n,k})$ has $2n=2|E|$ rows and $\binom{n}{k}$ columns of the form $(v, \mathbf{1}-v)$ where $v\in\{0,1\}^n$ is a vector with exactly $k$ ones, hence can be recognized in polynomial time (in its size). We denote such a matrix by $S_{n,k}$. If $k=1$, or equivalently $k=n-1$, $S=S_{n,1}=S_{n,n-1}$ is just the identity matrix $I_n$. The case $k=0$ or $k=n$ corresponds to a non-connected matroid whose base polytope is just a single vertex, and can be ignored for our purposes.

Finally, we remark that the connections between $1$- and $2$- sum of matroids and $1$- and 2-\op{} of (slack) matrices does not carry over to $k$-sum and $k$-\op{}, in particular the 3-sum of matroids does not seem to have a ``simple'' interpretation in terms of slack matrices (see \cite{aprile2019regular} for details about the polyhedral aspects of the 3-sum of matroids).


Before proceeding with some structural lemmas, we need some preliminary assumptions. Let $M(E,\calB)$ be a matroid such that $B(M)$ is 2-level, and let $S$ be a slack matrix of $B(M)$. From now on we assume that:
\begin{enumerate}
\item $S$ is 0/1.
\item \label{assump:loops} $M$ does not have loops or coloops. 
\item \label{assump:const} $S$ does not have any constant row (i.e., all zeros or all ones).

\end{enumerate}

Assumption \ref{assump:loops} is without loss of generality as, if $e$ is a loop or coloop of $M$, then $B(M)$ has a constant coordinate in correspondence of $e$ and is thus isomorphic to $B(M\setminus e)$. Similarly, Assumption \ref{assump:const} is without loss of generality as constant rows correspond to redundant inequalities and can always be removed from a slack matrix.

Finally, only for the sake of proving Lemmas \ref{lem:1-sum_matroid}, \ref{lem:2-sum_matroid}, we also make the following assumption throughout Sections \ref{sec:1sum}, \ref{sec:2sum}.

\begin{enumerate}\setcounter{enumi}{3}
\item\label{assump:nonneg} $S$ has a row for each inequality of the form $x_e\geq 0$ (we refer to such rows as non-negativity rows) and a row for each inequality $1-x_e\geq 0$ for $e\in E$.
\end{enumerate}

We now justify Assumption \ref{assump:nonneg}. 
One can show (using well-known facts from~\cite{tutte1965lectures}, see \cite{aprile2018thesis} for more details)
that for each element $e\in E$ at least one of the inequalities $x_e\geq 0$, $x_e\leq 1$ is facet defining for $B(M)$. 



Notice that these inequalities correspond to opposite faces of $B(M)$, and to complementary 0/1 rows in the slack matrix $S$ of $B(M)$ (two 0/1 rows are complementary if their sum is the all ones vector $\mathbf{1}$).
We remark that, for any slack matrix $S$ with a 0/1 row $r$, adding the complementary row $\mathbf{1}-r$ to $S$ still gives a slack matrix, which is a $k$-\op{} if and only if $S$ is. Hence we can assume that our slack matrix $S$ contains all the non-negativity rows.

\subsection{1-sums}\label{sec:1sum}

We now focus on the relationship between 1-sums and 1-\op{}s. As already remarked, if $S_1, S_2$ are slack matrices of $B(M_1), B(M_2)$ respectively, then $S_1\otimes S_2$ is a slack matrix of $B(M_1)\times B(M_2)=B(M_1\oplus M_2)$. We now show that the converse holds: this will make sure that, whenever we decompose the slack matrix of a matroid base polytope as a 1-\op, the factors are still matroid base polytopes.

\begin{lem} \label{lem:1-sum_matroid}
Let $M$ be a matroid and let $S$ be a slack matrix of $B(M)$. If $S=S_1\otimes S_2$ for some matrices $S_1, S_2$, then there are matroids $M_1, M_2$ such that $M=M_1\oplus M_2$ and $S_i$ is a slack matrix of $B(M_i)$ for $i \in [2]$.
\end{lem}
 \begin{proof} 
By Assumption \ref{assump:nonneg}, $S$ contains all the rows corresponding to inequalities $x_e\geq 0$, for each $e$ element of $M$. Each such non-negativity inequality belongs either to $S_1$ or to $S_2$, hence we can partition $E$ into $E_1, E_2$ accordingly. Recall that the row set of $S$ can also be partitioned into sets $R_1, R_2$, since each row of $S$ corresponds to a row of $S_1$ or $S_2$. Notice that $E_1 \neq \emptyset$; otherwise, all the rows corresponding to $x_e\geq 0$ belong to $R_1$.  Similarly, $E_2 \neq \emptyset$. But now the slack of a vertex of $B(M)$ with respect to every other inequality (of form $x(U)\leq \rk(U)$ for some $U\subseteq E$) depends entirely on the slack with respect to the rows in $R_1$, implying that each column of the matrix $S$ restricted to rows of $R_1$ can be completed to a column of $S$ in a unique way. Hence, since $S$ is a 1-\op, we must conclude that $S_2$ is made of a single column, contradicting the fact that $S$ does not have constant rows (Assumption \ref{assump:const}). 

Now, let $\calB_i=\{B\cap E_i: B\in \calB\}$ for $i \in [2]$. By definition of 1-\op{} of matrices, $B(M)=\{B_1\cup B_2: B_i\in \calB_i $ for $i \in [2]\}$. 
This implies that $M=M_1\oplus M_2$ where $M_i=M\restr{E_i}$  for $i \in [2]$, thus $B(M)=B(M_1)\times B(M_2)$. Hence, for every row of $S$ corresponding to an inequality $x(U)\leq \rk(U)$, we have either $U\subseteq E_1$, $U\subseteq E_2$, or the inequality is redundant and can be removed. In the first case, clearly the row is in $R_1$ as its entries depend only on the rows $x_e\geq 0$ for $e\in E_1$, and similarly in the second case the row is in $R_2$. Since removing redundant rows does not change the polytopes for which $S,S_1,S_2$ are slack matrices, we conclude that $S_i$ is a slack matrix of $B(M_i)$ for $i \in [2]$.
\end{proof}

\begin{cor}\label{cor:matroid_conn}
Let $M$ be a matroid and let $S$ be a slack matrix of $B(M)$. Then $M$ is connected if and only if $S$ is irreducible. 
\end{cor}

\subsection{2-sums}\label{sec:2sum}

Now, we deal with slack matrices of connected matroids and with the operation of 2-\op. 

We will need the following result, which provides a description of the base polytope of a 2-\op{} $M_1\oplus_2 M_2$ in terms of the base polytopes of $M_1, M_2$. Its proof can be derived from \cite{Grande16}, or found in \cite{aprile20182}.

\begin{lem}\label{obs:2sumpolydescription}
	Let $M_1, M_2$ be matroids on ground sets $E_1, E_2$ respectively, with $E_1\cap E_2=\{p\}$, and let $M=M_1\oplus_2 M_2$. Then $B(M)$ is affinely isomorphic to $$(B(M_1)\times B(M_2))\cap\{(x,y) \in\R^{E_1} \times \R^{E_2} \mid x_p+y_p=1\}.$$
\end{lem}

Lemma \ref{obs:2sumpolydescription} implies that, if $M=M_1\oplus_2 M_2$, then $B(M)$ is affinely isomorphic to the glued product of polytopes $B(M_1)$ and $\{(\overline{y}_p, y)\in \R\times \R^{E_2}: y\in B(M_2), \overline{y}_p=1-y_p\}$ with respect to variables $x_p$, $\overline{y}_p$ respectively. Notice that the latter polytope is affinely isomorphic to $B(M_2)$. 
Notice that, if $x_p\geq 0$ and $y_p\leq 1$ are facet-defining for $B(M_1)$, $B(M_2)$ respectively, then the glued product is indeed simplicial. Otherwise, thanks to Assumption \ref{assump:nonneg} above and to its justification, one checks that Corollary \ref{cor:glued} holds as well for this glued product.

Hence in both cases we have, by Corollary \ref{cor:glued} 
 that if $S_i$ is a slack matrix of $B(M_i)$ for $i \in [2]$, then $(S_1,x_p) \otimes_2 (S_2,\overline{y}_p)$ is a slack matrix of $B(M)$, where $x_p$ is the row corresponding to $x_p\geq 0$, $\overline{y}_p$ the row corresponding to $y_p\leq 1$, and $\overline{y}_p$ is a row of $S_2$ by assumption (iv). 

 The only missing ingredient is now a converse to the above statement.

\begin{lem} \label{lem:2-sum_matroid}
Let $M = (E,\calB)$ be a connected matroid and let $S$ be a slack matrix of $B(M)$. Assume there are $S_1, S_2$ such that $S=(S_1,x_1)\otimes_2 (S_2,\overline{y}_1)$, for some rows $x_1, \overline{y}_1$, and let $S_1'=S_1+(\mathbf{1}-x_1)$ and similarly for $S_2'$. Then there are matroids $M_1, M_2$ such that $M=M_1\oplus_2 M_2$ and $S_i'$ is a slack matrix of $B(M_i)$, for $i=1,2$.
\end{lem}

\begin{proof}
 First, we claim that the special row $r$ of $S$ does not correspond to any non-negativity inequality (which are all present in $S$ thanks to Assumption \ref{assump:nonneg}). Indeed, if it corresponds to $x_e\geq 0$ for some $e\in E$, then it is not hard to see that $S[0]$ is the slack matrix of $M\setminus e$, and similarly $S[1]$ is the slack matrix of $M/e$ (where $S[0], S[1]$ are defined as in Section \ref{sec:intro}). But both matrices are 1-\op{}s, hence by Corollary \ref{cor:matroid_conn}, none of $M\setminus e, M/e$ is connected. But this is in contradiction with the well-known fact that, if $M$ is connected, then at least one of $M\setminus e, M/e$ is connected (see~\cite{tutte1965lectures}).
 
 Hence, the special row $r$ corresponds to an inequality $x(U)\leq \rk(U)$ for some $U\subseteq E$, while each inequality $x_e\geq 0$ corresponds to a row of either $S_1$ or $S_2$, giving a partition of $E$ in $E_1, E_2$. We will now proceed similarly as in the proof of Lemma \ref{lem:1-sum_matroid}. First, by noticing that the slack of any vertex with respect to $x(U)\leq \rk(U)$ depends exclusively on the slack with respect to the non-negativity inequalities, we can again conclude that $E_1, E_2$ are not empty.
 
 Moreover, we claim that $\rk(E_1)+\rk(E_2)=\rk(E)+1$: this, thanks to Proposition \ref{prop:sumrank}, implies that $M=M_1\oplus_2 M_2$ where $M_i$ is a matroid on ground set $E_i\cup\{p\}$, for $i=1,2$.
  Notice that $\rk(E_1)+\rk(E_2)\geq \rk(E)+1$, since $M$ is connected (see again Proposition \ref{prop:sumrank}).
 
 Let us partition the bases of $M$ (hence the columns of $S$) into $\calB_0$ and $\calB_1$ according to the row $r$ and its corresponding inequality: in particular, $B\in \calB_0$ if $|B\cup U|=\rk(U)$ and $B\in \calB_1$ if $|B\cup U|=\rk(U)-1$. Let $\calB_0^i=\{B\cap E_i: B\in \calB_0\}$ for $i=1,2$: since
 bases of $\calB_0$ correspond to columns of $S[0]=S_1[0]\otimes S_2[0]$, we have that $B\in \calB_0$ if and only if $B=B_1\cup B_2$ with $B_i\in \calB_0^i$. But this implies that the (independent) sets in $\calB_0^1$ (resp. $\calB_0^2$)  have all the same size $r_0^1$ (resp. $r_0^2$). The same can be argued for bases in $\calB_1$, by defining $\calB_1^i$, $r_1^i$, $i=1,2$ analogously. Notice that $r_0^1+r_0^2=r_1^1+r_1^2=\rk(M)$. But then $\rk(E_i)=\max(r_0^i, r_1^i)$, $i=1,2$. Notice that, if $r_0^i=r_1^i$ for some $i\in\{0,1\}$, then $\rk(E_1)+\rk(E_2)=\rk(E)$, a contradiction. Hence, without loss of generality, we must have $\rk(E_1)=r_0^1>r_1^1$ and $\rk(E_2)=r_1^2>r_0^2$. Notice that our claim follows from showing that $r_1^1=r_0^1-1$. If, by contradiction, $r_1^1>r_0^1-1$, then we would have that all bases $B$ of $M$ satisfy either $|B\cap E_1|=r_1^1$ or $|B\cap E_1|\leq r_1^1-2$, but this for instance contradicts the bijective basis exchange axiom (see Exercise 12 of Chapter 12 of \cite{oxley2006matroid}, or \cite{brualdi1969very}).
 
 Now, from Lemma \ref{obs:2sumpolydescription} and the subsequent discussion we have that $S$ is the 2-product of the slack matrices of $B(M_1)$, $B(M_2)$: by arguing similarly as in the proof of Lemma \ref{lem:1-sum_matroid}, we conclude that such matrices are exactly $S_1$, $S_2$ (or $S_1'$, $S_2'$).
\end{proof}

\subsection{Concluding the proof}
We are now ready to prove the main result of this section, namely, Theorem \ref{thm:main3}. We recall that, as observed in Section \ref{sec:def}, we can assume that our input $S\in \{0,1\}^{m\times n}$ is preprocessed such that, if $S$ is a slack matrix, it is non-redundant.



\begin{proof}[Proof of Theorem~\ref{thm:main3}]
We first check whether $S$ is isomorphic to $S_{d,k}$ for some $d$ and $k$, in which case we are done.

Then, we run the algorithm from Theorem \ref{thm:main1} with $k=1$ and, if $S$ is a 1-\op, we decompose it into irreducible factors $S_1,\dots, S_t$ and test each $S_i$ separately. By Lemma \ref{lem:1-sum_matroid} we have that $S$ is a slack matrix of $B(M)$ if and only if $S_i$ is a slack matrix of $B(M_i)$ for each $i \in [t]$, and $M=M_1\oplus \dots \oplus M_t$.

We can now assume that $S$ is irreducible. We run the algorithm from Theorem \ref{thm:main1} for $k=2$ to iteratively decompose $S$ as a 2-\op{}. Notice that we ignore decompositions as $S=S'\otimes_2 I_2$, where $I_2$ is the $2\times 2$ identity matrix and only decompose $S$ when both factors are different from $I_2$ (see the discussion at the end of Section \ref{sec:def}). We decompose $S$ as the 2-\op{} of matrices $S_1,\dots,S_t$ where $S_i=S_{d_i,k_i}$ for some $d_i\geq k_i\geq 1$ and for $i \in [t]$. If this is not possible (i.e. one of the factors $S_i$ is not of the desired form, and cannot be decomposed further), then thanks to Lemma \ref{lem:2-sum_matroid} we conclude that $S$ is not a slack matrix of a base polytope. 

Now, there is one last technicality we have to deal with, before we can apply Lemma \ref{obs:2sumpolydescription} and conclude that $S$ is the slack matrix of a base polytope. We need to ensure that each pair of special rows involved in a 2-\op{} is ``coherent'', i.e. we need that, for each 2-sum between matroids $S_i$, $S_j$, the special row of one of them corresponds to $x_p\geq 0$, for some element $p$, and the special row of the other corresponds to $x_p\leq 1$. If this does not happen, i.e. if, say, both rows are of form $x_p\geq 0$, then the matrix resulting from the 2-product is the slack matrix of a polytope which is not the base polytope of a matroid. 
Note that, unless $S_i$ is the identity matrix (in which case all its rows are non-negativity rows), we can choose whether $S_i$ is the slack matrix of $U_{d_i,k_i}$ or of its dual $U_{d_i,d_i-k_i}$, hence we can choose the form of the special row. Hence $S$ is the slack matrix of a matroid polytope if and only if there is a choice that makes all the pairs of special rows coherent. This problem can be easily solved as follows. Define a tree with nodes $S_1,\dots,S_t$, where two nodes $S_i, S_j$ are adjacent if the 2-\op{} $S_i\otimes_2 S_j$ occurs during the decomposition of $S$. Now, by coloring the nodes of the tree by two colors, according to the form of the special row, one can efficiently determine whether there exists a proper coloring satisfying the ``fixed'' colors (given by the $S_i$'s that are identity matrices). Notice that, if there exists a feasible coloring, then this determines a matroid $M$, and it is essentially unique: it is easy to see that the only other possible coloring gives rise to the dual matroid $M^*$, corresponding to the same slack matrix. This concludes the algorithm. Notice that, in case $S$ is the slack matrix of $B(M)$, $M$ (or its dual) can be reconstructed by successively taking the 2-\op{} of the $U_{d_i,k_i}$'s (or of their duals, depending on the coloring found).
\end{proof}

\section{Recognizing stable set polytopes of perfect graphs}\label{sec:stab(G)k-sum} 
We now turn to the second application of our results to a class of combinatorial polyotopes. 
First, we give an interpretation of the $k$-\op{} operation in the context of slack matrices of stable set polytopes. For a graph $G$, we denote its stable set polytope\footnote{The stable set polytope of a graph is the convex hull of its stable sets.} by $\stab(G)$.

We say that a graph $G(V,E)$ has a \emph{clique cut-set} $K$ if $V$ can be partitioned in $V_1, V_2, K$ such that $V_1,V_2\neq \emptyset$, $K$ is a clique, and there is no edge of $G$ between $V_1, V_2$. Note that we allow $K$ to be empty: this is equivalent to $G$ being disconnected, in particular $G$ is the disjoint union of $G_1, G_2$ with vertex sets $V_1, V_2$ respectively. It is easy to see that in this latter case $\stab(G)=\stab(G_1)\times \stab(G_2)$, hence if $S, S_i$ denote the slack matrices of $\stab(G), \stab(G_i)$ for $i=1,2$, we have $S=S_1\otimes S_2$. 
 We now remark that a generalization of this fact holds. We recall the following theorem of Chv{\'a}tal \cite[Theorem 4.1]{chvatal1975certain}, which we rewrite in a convenient form.
 
 \begin{thm}\label{thm:cliquecutset}
 Let $G(V,E)$ be a graph and $(V_1,V_2,K)$ a partition of $V$, where $K$ is a clique cut-set. For $i=1,2$, let $G_i:=G[K \cup V_i]$. Then $\stab(G)$ is affinely isomorphic to the simplicial glued product of $\stab(G_1)$ and $\stab(G_2)$ with respect to coordinates $x_v$, for $v\in K$.
 \end{thm}


The following is a consequence of Lemma \ref{lem:k-sum_slack}, Theorem \ref{thm:cliquecutset}, and the fact that the properties of simplicial glued product are easily verified in this case.  

\begin{thm}\label{thm:stab(G)k-sum}
Let $G(V,E)$ be a graph and $S$ be a slack matrix of $\stab(G)$, and let $k\geq 1$. Then the following are equivalent:
\begin{enumerate}
    \item[(i)]
$S = (S_1,x_1,\dots,x_{k-1}) \otimes_k (S_2,y_1,\dots,y_{k-1})$ where for $i=1,2$, $S_i$ is a slack matrix of the stable set polytope of a graph $G_i$, and $x_1,\dots,x_{k-1}$ (resp.~$y_1,\dots,y_{k-1}$) correspond to certain nonnegativity inequalities in $G_1$ (resp.~$G_2$); 
\item[(ii)] $G$ has a clique cut-set of size $k-1$.
\end{enumerate}
\end{thm}

The facial description of the stable set polytope of a general graph can be rather complicated. However, the situation is different when the graph $G$ is perfect: then
$\stab(G)$ is a 2-level polytope whose slack matrix is well understood (see \cite{aprile2020extended,aprile2017extension,yannakakis1991expressing}). We recall (see \cite{chvatal1975certain}) that a graph $G(V,E)$ is perfect if and only if
$$\stab(G)=\{x\in\R_+^V: x(C)\leq 1 \text{ for all maximal cliques $C$ of }G\}.$$ 
Moreover, this description is non-redundant. Hence the 0/1, non-redundant slack matrix of $\stab(G)$ will have a column for each stable set of $G$ and a row for each non-negativity inequality and for each maximal clique of $G$ (we call such rows and inequalities \emph{clique rows} and \emph{clique inequalities}, respectively). We now describe a polynomial-time algorithm to recognize slack matrices of stable set polytopes of perfect graphs. Interestingly, this does not rely on the notion of $k$-\op{}. 


\begin{proof}[Proof of Theorem \ref{thm:recog-slack-stab(G)}]
 First, as argued in Section \ref{sec:def}, we can restrict ourselves to the case of non-redundant slack matrices. Given $S\in \{0,1\}^{m\times n}$, let the rank of $S$ be $r$, hence if $S$ is a slack matrix of a polytope $P$, the dimension of $P$ is $d=r-1$ thanks to Lemma \ref{lem:slackconvcol}. 
 We know that if $S$ is the slack matrix of $\stab(G)$ (we refer to this as the YES case), for a graph $G$ with $d$ vertices, then there is a column of $S$ (corresponding to the empty stable set) with exactly $d$ zeros. Hence, if no such column exists, we output NO. 
 
 Otherwise, for each column $c$ with exactly $d$ zeros, we proceed as follows. If this is a YES case and $c$ corresponds to the empty stable set, the $d$ rows that have zeros in column $c$ must correspond to the non-negativity inequalities, and the other rows must correspond to clique inequalities. This allows us to construct a graph $G$ by connecting two vertices if they do not form a stable set, i.e., if the two corresponding rows do not have a 1 in the same position. Now, the columns of $S$ give us a list $L_1$ of stable sets of $G$, some of which are singletons. By looking at the rows of $S$ that correspond to clique inequalities, restricted to the singleton columns of $S$, we obtain a list $L_2$ of subsets of vertices of $G$ that are supposed to be maximal cliques of $G$. 
 
 We need to check that $L_2$ is indeed the list of all maximal cliques of $G$, and, similarly, that $L_1$ contains all the stable sets of $G$ (we just write that $L_1, L_2$ are \emph{correct}). Now, it is well-known that the maximal cliques (equivalently, maximal stable sets) of a graph can be enumerated in total polynomial time (i.e., in time polynomial in the size of the output), for instance using the Bron–Kerbosch algorithm (\cite{bron1973algorithm}). This allows to efficiently check whether $L_2$ is correct and whether the maximal sets in $L_1$ are indeed the maximal stable sets of $G$.
 
 Now, let $n'$ be the number of (not necessarily maximal) stable sets of $G$.
 We remark that $L_1$ is correct if and only if $n'$ is equal to $n$, the number of columns of $S$. Hence, we only need to argue that we can efficiently establish whether $n'=n$ (without explicitly computing $n'$, that might be much larger than $n$).
 Let $S_1,\dots, S_t$ be the maximal stable sets of $G$ computed by the Bron–Kerbosch algorithm. Then we have $n'=|\{S\subseteq V(G): S\subseteq S_i$ for some $i\in[t]\}|$. Notice that a formula for $n'$ is given by the inclusion-exclusion principle: indeed, given power sets $2^{S_1}$, $\dots$, $2^{S_t}$, $n'$ is exactly the size of the union of such sets. 
 
 In order to check whether $n=n'$, we proceed as follows. First, assume without loss of generality that $S_1$ is the largest maximal stable set of $G$, then we can check whether $2^{|S_1|}\leq n$, otherwise we stop and output NO.
 Then we proceed by simply writing down all subsets of $S_1$, then all subsets of $S_2$ that we have not already written down, etc., until either we listed $n+1$ sets (hence we stop and output NO) or we compute $n'\leq n$. Notice that going through all such subsets takes time proportional $\sum_{i=1}^t 2^{|S_i|}\leq tn \leq n^2$. Hence, the whole process takes polynomial time in $n$, and at the end we can establish whether $n=n'$.

 Hence we can efficiently check whether the lists are correct in time polynomial in their size, i.e., polynomial in the size of $S$. If all these checks are successful, then we only need to check that the entries of clique rows of $S$ are correct: for each such row, and for each column of $S$ corresponding to a stable set, the corresponding entry of $S$ must be 1 if and only if the clique and the stable set intersect, 0 otherwise. Finally, if this check is also successful then $S$ is the slack matrix of $\stab(G)$ (with our selected column $c$ corresponding to the empty stable set) and we output YES. 
 
 If the check fails for each choice of the column $c$ as above, we output NO. In the worst case, we need to iterate the above procedure over all columns of $S$ (with exactly $d$ zeros) and the thesis follows. 
\end{proof}

\section*{Acknowledgements}

This project was supported by ERC Consolidator Grant 615640-ForEFront, by a gift by the SNSF and by the ONR award N00014-20-1-2091.

\printbibliography

\end{document}